\documentclass[copyright,creativecommons]{eptcs}
\providecommand{\event}{SYNT 2016} 

\usepackage[utf8]{inputenc}
\usepackage[mathscr]{euscript}
\usepackage{graphicx}
\usepackage{amssymb}
\usepackage{amsmath}
\usepackage{amsthm}

\usepackage{algorithm}
\usepackage{algorithmic}

\usepackage[usenames,dvipsnames]{color}

\usepackage[htt]{hyphenat}
\usepackage{url}
\usepackage{breakurl} 

\usepackage{hyperref}
\providecommand{\doi}[1]{\textsc{doi}: \href{http://dx.doi.org/#1}{\nolinkurl{#1}}}

\usepackage{xspace}
\usepackage{color}

\usepackage{listings}
\definecolor{DarkGreen}{rgb}{0,0.3,0}
\definecolor{DarkBlue}{rgb}{0,0,0.7}
\lstdefinelanguage{SMV}[]{}{
  commentstyle=\color{DarkGreen}\itshape,
  keywordstyle=\color{blue}\bfseries,
  morekeywords=[1]{AUXVAR,VAR,INIT,LTLSPEC,TRANS,GUARANTEE,ASSUMPTION,TRUE,X,G,F,CTLSPEC,AG,EX,MODULE,VARENV,DEFINE,LTLSPECENV,next,new,boolean,WEIGHT,abs,PREV},
  morekeywords=[1]{occurs,between,and,GF,leads,to,After,have,at,most,two,until,Globally,S,responds,
  after}, morecomment=[l]{--}
}

\lstdefinelanguage{MAA}[]{}{
  morekeywords=[1]{state,initial},
  morecomment=[l]{//}
}

\newenvironment{packed_itemize}{
\begin{itemize}
  \setlength{\itemsep}{1pt}
  \setlength{\parskip}{0pt}
  \setlength{\parsep}{0pt}
}{\end{itemize}}

\newtheorem{theorem}{Theorem}
\newtheorem{lemma}{Lemma}

\newcommand{\myTitle}{Symbolic BDD and ADD Algorithms for Energy Games}

\title{\myTitle}

\author{Shahar Maoz \qquad Or Pistiner \qquad Jan Oliver Ringert
\institute{School of Computer
Science\\
Tel Aviv University, Israel}
}

\def\titlerunning{\myTitle} 
\def\authorrunning{S. Maoz, O. Pistiner \& J.O. Ringert}

\newcommand{\op}[1]{\textbf{\texttt{#1}}\xspace}

\begin{document}

\maketitle

\begin{abstract}
Energy games, which model quantitative consumption of a limited
resource, e.g., time or energy, play a central role in quantitative models for reactive
systems. Reactive synthesis
constructs a controller which satisfies a given specification, if one
exists. 
For energy games a
synthesized controller ensures to satisfy not only the safety constraints of the
specification but also the quantitative constraints expressed in the
energy game. A symbolic algorithm for energy games, recently 
presented by Chatterjee et al.~\cite{ChatterjeeRR14}, is symbolic in
its representation of quantitative values but concrete in the representation of 
game states and transitions. 
In this paper we present an
algorithm that is symbolic both in the quantitative
values and in the underlying game representation. We have
implemented our algorithm using two different symbolic
representations for reactive games, Binary Decision Diagrams (BDD) and 
Algebraic Decision Diagrams (ADD). We investigate the
commonalities and differences of the two implementations and compare their
running times on specifications of energy games.
\end{abstract}

\section{Introduction}

Reactive synthesis is an automated procedure to obtain a
correct-by-construction reactive system from its temporal logic
specification~\cite{PR89}.
Rather than manually constructing a system and using model checking to
verify its compliance with its specification, synthesis offers an
approach where a correct implementation of the system is automatically
obtained, if such an implementation exists. Traditionally,
specifications for synthesis express qualitative properties of desired
system behavior, which are either satisfied or not satisfied.
Over the last decade interest has increased for quantitative
properties~\cite{BloemCHJ09}, which can
express, e.g., cost of actions or probability of success, and allow for
synthesis of optimal solutions.

One intuitive and popular quantitative extension of games for reactive
synthesis are energy games (EGs), defined by Bouyer et al.
\cite{BouyerFLMS08}, which model quantitative consumption of a
limited resource, e.g., cost, time, or energy. Transitions between
states are annotated with weights. A play 
starts with a finite energy level that is updated by the weight on every
transition. An infinite play is winning for the system if the energy
level never goes below 0. The EG is realizable if the system
has a strategy to win for initial choices of the environment.
The objective of synthesis for EG is to construct such a
strategy.

Brim et al.~\cite{BrimCDGR11} have presented an efficient
pseudo-polynomial algorithm for solving EGs (polynomial in the
state space and the maximal weight). The algorithm is bounded by the
maximal initial energy. Because the maximal initial energy for
realizable energy games is bounded the algorithm is complete.
It computes the minimal energy required to win from any state in a
backwards manner, using a fixed point
calculation. A representation by minimal energy levels is symbolic in
the sense of antichains~\cite{DoyenR10}, also used in the algorithm of
Chatterjee et al.~\cite{ChatterjeeRR14}. These algorithms are thus
symbolic in the quantitative values but concrete in the representation
of game states and transitions\footnote{Note that the implementation
presented by Chatterjee et al.~\cite{ChatterjeeRR14} avoids a
full concretization of underlying safety games by using antichains as
described in~\cite{FiliotJR09}.}.

\textbf{In this work we present a novel algorithm that is optimized for
reactive energy games and symbolic both in the quantitative values and
in the underlying game representation.} Our algorithm implements a fixed
point computation similar to Chatterjee et al.~\cite{ChatterjeeRR14} but
due to the reactive nature of the game it updates energy levels of
system and environment states in one instead of two steps. Our symbolic
algorithm also uses antichains. However, the antichains are now defined
over sets of states and their sets of energy levels instead of over
single states and their sets of energy levels. These modifications
require efficient symbolic data structures and corresponding operations
to represent the (intermediate) results of the fixed point computation.

We further present two different implementations of our algorithm. The first
implementation is based on Binary Decision Diagrams
(BDDs)~\cite{Bryant86} to represent states, transitions, and weight
definitions. Antichains are then maps mapping minimal energy levels to
BDDs. Accordingly, the fixed point calculation of the algorithm is based
on iterations over minimal energy levels and weights and is executed using
symbolic BDD operations.

Our second implementation is based on Algebraic Decision Diagrams
(ADDs)~\cite{BaharFGHMPS97}, a special case of Multi-Terminal Binary
Decision Diagrams (MTBDDs)~\cite{FujitaMY97}. ADDs have numbers as
terminal nodes --- instead of \texttt{TRUE} and \texttt{FALSE} in the
case of BDDs --- and provide efficient implementations of symbolic
algebraic operations. With ADDs, an antichain in our algorithm is
expressed in a single ADD. Again, we have implemented the fixed point
calculation of the algorithm using only symbolic ADD operations.

The ADD and BDD algorithms do not only use symbolic data structures
for efficient representation:
\textbf{contributions of our implementations are also their specific use
of symbolic manipulations for efficiently performing quantitative
operations.} We explain both implementations to highlight their commonalities and
differences. In Sect.~\ref{sec:evaluation} we compare the running times
of the BDD and ADD implementations against each other and against
themselves over increasing sizes of specifications and different
specifications of weights.

We present background on games, BDDs, ADDs, and EGs in
Sect.~\ref{sec:preliminaries}. Sect.~\ref{sec:example} introduces a
running example specification. We present our generic algorithm in
Sect.~\ref{sec:algorithm},  and its BDD and ADD implementations in
Sect.~\ref{sec:algorithmBDD} and Sect.~\ref{sec:algorithmADD} resp. 
Sect.~\ref{sec:evaluation} presents a preliminary
evaluation of our two implementations. We discuss the results and
related work in Sect.~\ref{sec:discussion} and conclude in
Sect.~\ref{sec:conclusion}.

\section{Preliminaries}
\label{sec:preliminaries}

\subsection{Infinite Games, BDDs, ADDs}

We repeat some basic definitions of games, BDDs, and
ADDs. We also describe the general approach for
symbolically representing games between an environment and a system
player using BDDs.

\paragraph{Games, game graphs, plays, and strategies}
We consider infinite games played between two players on a finite
directed graph as they move along its edges. For a game we define a \textit{game
graph} as a tuple $\Gamma = \langle G=\langle V,E,w\rangle, V_{0}, V_{1}\rangle$,
where $G=\langle V,E,w \rangle$ is a finite directed
graph with a weight function $w:
E \rightarrow \mathbb{Z}$ that attaches weights to its edges. $V_{0},
V_{1}$ is a partition of $V$ into $V_{0}$, the set of player-0 (the maximizer)
vertices, and $V_{1}$, the set of player-1 (the minimizer) vertices. Each
vertex $v \in V$ has out degree at least one, i.e., $G$ has no
deadlocks.
A play starts by placing a pebble on a given initial vertex, and
continues infinitely many rounds as the two players move the pebble
along the edges of $G$.
In each round, if the pebble is on a vertex $v\in{V_i},~i\in\{0, 1\}$,
then player-$i$ chooses an outgoing edge from $v$ to some adjacent
vertex $u\in{V}$, and the next round starts with the pebble on $u$. The
infinite path formed by the rounds is called a \textit{play}.
A \textit{strategy} of player-$i$ is a function that given the prefix of
a play ending in a vertex of player-$i$ returns a successor vertex.
\paragraph{Binary Decision Diagrams and Algebraic Decision Diagrams}
Binary decision diagrams (BDDs)~\cite{Bryant86} are a compact data
structure for representing and manipulating Boolean functions. For the
purpose of verification and synthesis they are usually used to represent
sets of Boolean vectors~\cite{HuthR04}, e.g., encoding sets of states or
transitions.
Bryant~\cite{Bryant86} showed how Boolean operations on BDDs can be
efficiently implemented, including logical connectives and existential
and universal abstractions. ADDs~\cite{BaharFGHMPS97} are a
generalization of BDDs, such that the terminal nodes may take on values
belonging to a set of constants different from 0 and 1, such as integers
or reals. Another name for ADDs is Multi-Terminal Binary Decision
Diagrams (MTBDDs)~\cite{FujitaMY97}, which reflects their structure
rather than their application for computations in algebras. An ADD
represents a Boolean function of $n$ variables
$f:\{0,1\}^{n}\rightarrow{S}$ where $S$ is a set of constants. Boolean,
arithmetic, and abstraction operations are all applicable to ADDs. We
now briefly state the most important operations that we use in our ADD
based algorithm. For ADDs $g$, $h$, and a $0$-$1$ ADD $f$, arithmetic
operators $\texttt{op} \in \{+, -, \cdot, max, \ldots\}$, denoted
$g~\texttt{op}~h$, operate on the terminal nodes for common variable
assignments, the \textit{If-then-else} Boolean operation is defined as
$\texttt{ITE}(f,g,h) = f\cdot{g}+\neg f\cdot{h}~$, and abstraction of
variables $v$ from $g$, which aggregate terminal values of $g$ for all
assignments to variables in $v$ by operators $min$ or $max$.
\paragraph{Symbolic representation of game graphs for reactive games}
Reactive games are turn-based two players games, between an environment
and a system player, in which the environment always plays first and the
system \textit{reacts}. The environment controls the input variables,
denoted by $env$, and the system controls the output variables, denoted
by $sys$, all are assumed to be Boolean, such that their values are
modified in each step of the game~\cite{PR89}. We denote by
$\mathscr{V}_{var}:=\{0,1\}^{var}$ all the possible assignments to the
variables $var$, where $s_{var}\in{\mathscr{V}_{var}}$ is some
assignment to $var$. We now define the symbolic weighted game graph
$G^{sym} := \langle \theta^e,\theta^s,\rho^e,\rho^s, w\rangle$ for
reactive games.
A \textit{state} $s$ in $G^{sym}$ is an assignment to \textit{all}
variables, i.e., $s := s_{env\cup{sys}}\in{\mathscr{V}_{env\cup{sys}}}$.
That means each state belongs to both players (unlike the general
model).
We denote by $\theta^e, \theta^s$ the sets of \textit{initial} states of
the environment and the system, respectively. Each is represented by a
BDD (or by a 0-1 ADD) that encodes its characteristic function:
$\chi_{\theta^e}:\mathscr{V}_{env}\rightarrow\nolinebreak\{0,1\}$ and
$\chi_{\theta^s}:\mathscr{V}_{env\cup{sys}}\rightarrow\{0,1\}$.
Each player has a \textit{transition relation} that defines its valid
transitions in $G^{sym}$. It is denoted by $\rho^e$ and $\rho^s$ for the
environment and the system, respectively. A next state in $G^{sym}$, as
opposed to a current state, is represented by a \textit{primed} version
of the variables $sys'$ and $env'$. Therefore, the
BDDs (or 0-1 ADDs) that represent $\rho^e$ and $\rho^s$ encode the
characteristic functions
$\chi_{\rho^e}:\nolinebreak\mathscr{V}_{env\cup{sys}\cup{env'}}\rightarrow\nolinebreak\{0,1\}~\text{and}~
\chi_{\rho^s}:\nolinebreak\mathscr{V}_{env\cup{sys}\cup{env'}\cup{sys'}}\rightarrow\{0,1\}$.
The bijection $prime:
\mathscr{V}_{sys \cup env} \rightarrow \mathscr{V}_{sys' \cup env'}$
replaces unprimed variables by their primed counterparts.
We denote a transition from state $s_1$ to state $s_2$ by
$t_{s_{1},s_{2}} := s_1 \cup prime(s_2)
\in{\mathscr{V}_{env\cup{sys}\cup{env'}\cup{sys'}}}$.
It is valid if $t_{s_{1},s_{2}}\in{\rho^e\cap{\rho^s}}$, thus consists
of a valid transition for environment choice $s^{e}_{2} :=
s_{env}\in{\mathscr{V}_{env}}$, denoted by $t_{s_{1}, s^{e}_{2}} :=
s_1\cup{prime(s^{e}_{2})}\in{\rho^e}$, followed by a valid transition
for system choice $s^{s}_{2} := s_{sys}\in{\mathscr{V}_{sys}}$, denoted by $t_{s_{1},s^{e}_{2},s^{s}_{2}} :=
s_1\cup{prime(s^{e}_{2})}\cup{prime(s^{s}_{2})}\in{\rho^s}$.
Also, $G^{sym}$ has a \textit{weight function}
$w:\mathscr{V}_{env\cup{sys}}\times{\mathscr{V}_{env\cup{sys}}}\rightarrow\mathbb{Z}\cup{\{\perp\}}$
that attaches weights to its transitions, such that for all pairs of
states $s_1, s_2$ if $t_{s_{1},s_{2}}\in{\rho^e\cap{\rho^s}}$, $w(s_1,
s_2)\in{\mathbb{Z}}$, and otherwise, $w(s_1, s_2) = \perp$. For details
of its actual representation see Sect.~\ref{sec:algorithmBDD} and
Sect.~\ref{sec:algorithmADD}.

\paragraph{Notation}
We work with abstractions where a BDD represents a set
of states or transitions, i.e., $BDD \equiv
\text{Set of States}$ or $BDD \equiv \text{Set of Transitions}$.  An ADD
represents a function assigning an integer, plus or minus infinity, to every state or
transition, i.e., $ADD \equiv \text{States} \rightarrow
\mathbb{Z}_{\pm\infty}$ or $ADD \equiv \text{Transitions} \rightarrow
\mathbb{Z}_{\pm\infty}$.

\subsection{Energy Games}

\textit{Energy games} (EG), also called the \textit{lower bound problem}, have
been studied by Bouyer et al.\cite{BouyerFLMS08}. We give the formal
definitions of EG as defined in~\cite{BouyerFLMS08, BrimCDGR11}, followed by a
fixed point formulation of the solution for EG as in~\cite{BouyerFLMS08,
ChatterjeeRR14}.
\paragraph{EG definition and objectives} We add to an infinite
play on an EG graph ${\Gamma}$ an \textit{initial credit} or
\textit{initial energy} value $c\in{\mathbb{N}}$. We define the
\textit{energy level} of the prefix $v_{0}v_{1}\ldots v_{j}$ of the play
$p=v_{0}v_{1}\ldots v_{j}\ldots$, by $\textsf{EL}(p,j) = c +
\sum_{i=0}^{j-1} w(v_{i},v_{i+1})$. The objective of player-0 is to
construct an infinite play $p=v_{0}v_{1}\ldots v_{j}\ldots$ such that
the energy level is always non-negative during the play, i.e.,
$\textsf{EL}(p,j)\geq{0}~for~all~j\geq{1}$. We say that such a play $p$
is \textit{winning} for player-0. Otherwise, if the energy level goes
below zero during $p$, it is winning for player-1. Given an initial
credit $c$, a strategy is \textit{winning} for player-0 from $v\in{V}$
for $c$ if all plays resulting from the strategy are winning for
player-0. A vertex $v$ is winning for player-$i$ if there exists a
winning strategy for player-$i$ from $v$ for some initial credit $c$.


We consider two EG problems: (1) \textit{Decision problem}: Given a
vertex $v$, the problem asks if there exists an initial amount of energy
that suffices for player-0 to win from $v$; and (2) \textit{Minimum
credit problem}: which asks for every vertex $v\in{V}$, what is the
\textit{minimal} initial amount of energy that suffices for player-0 to
win from $v$.
A solution of the second also provides a solution for the first by
checking whether the minimal initial energy is finite. In addition a
solution for the second allows for the construction of optimal
memoryless strategies for solving EGs~\cite{ChatterjeeRR14}. We thus
focus on the minimum credit problem.

\paragraph{Algorithm for solving energy games from~\cite{ChatterjeeRR14}}
Both Chatterjee et al.~\cite{ChatterjeeRR14} and Bouyer et
al.~\cite{BouyerFLMS08} present algorithms to solve the minimal credit
problem of EGs. Their solutions are presented as the greatest fixed point of
a monotone operator on a power set lattice. Given an EG graph $\Gamma =
\langle G=\langle States,E, w\rangle, States_{0}, States_{1}\rangle$ and a bound $c\in{\mathbb{N}}$, we denote by $SE(c)$
the pairs of states and energy values up to $c$, i.e.,
${States}\times{\{n\in\mathbb{N}\mid n\leq{c}\}}$. We relate to the
monotone bounded operator
$\textsf{Cpre}_{c}:\mathbb{P}{(SE(c))}\rightarrow{\mathbb{P}(SE(c))}$,
where $\mathbb{P}(SE(c))$ is the power set of $SE(c)$, that is defined
in~\cite[Eqn.
1]{ChatterjeeRR14}. Chatterjee et al.~\cite{ChatterjeeRR14} present a
symbolic fixed point algorithm that computes a set of pairs consisting
of a state and an energy level that suffices for player-0 to win
for that state. It performs iterative applications of
$\textsf{Cpre}_{c}$ to $SE(c)$, which results in a finite
$\subseteq$-descending chain whose last element approximates the
greatest fixed point of $\textsf{Cpre}_{c}$.

\paragraph{Symbolic representation of upward closed sets} We define a
\textit{partial order} $\preceq~\subseteq{SE(c)\times{SE(c)}}$ (i.e.,
reflexive, transitive, antisymmetric binary relation), such that for all
$(s,e), (s',e') \in{SE(c)},~{(s,e)\preceq{(s',e')}}$ iff
${s=s'\wedge{e\leq{e'}}}$.
We call $\textsf{Up}(\preceq,
A):=\{(s,e)\in{SE(c)}\mid\exists{(s',e')\in{A}}:(s',e')\preceq{(s,e)}\}$
the $\preceq\textit{-upward }\newline\textit{closure}$ of $A\subseteq{SE(c)}$.
A set $A\subseteq{SE(c)}$ is
$\preceq$-$\textit{upward-closed}$ if for all $(s,e),(s',e')
\in{SE(c)}$, if ${(s,e)\in{A}}$ and $(s,e)\preceq{(s',e')}$, then
$(s',e')\in{A}$. An $\preceq$-upward-closed set equals its upward
closure. We denote by $\textsf{Min}(\preceq,A)$ the \textit{minimal
elements} of $A$, formally
$\textsf{Min}(\preceq,A):=\{(s,e)\in{A}~\mid~\forall{(s',e')\in{A}}:(s',e')\preceq{(s,e)}\Rightarrow{(s',e')
= (s,e)}\}$.
A set $A\subseteq{SE(c)}$ is an \textit{antichain} if all pairs
$(s,e)\neq{(s',e')}\in{A}$ are $\preceq$-incomparable. An
$\preceq$-$\textit{upward-closed}$ set is symbolically represented by
its minimal elements, as the former is \textit{uniquely} determined by
the latter:
if $A\subseteq{SE(c)}$ is an upward closed set, then
${\textsf{Up}(\preceq, \textsf{Min}{(\preceq,A))} = A}$.

\paragraph{Description of the algorithm from~\cite{ChatterjeeRR14} by
operations on antichains}  Bouyer et al.~\cite{BouyerFLMS08} present an
alternative description to the fixed point solution in terms of
sufficient \textit{infimum credits}. Since every element $U$ of the
$\textsf{Cpre}_{c}$ application from ~\cite[Eqn.1]{ChatterjeeRR14} is
$\preceq$-$\textit{upward-closed}$, it can be symbolically represented
by the \textit{antichain}
$\textsf{Min}(\preceq,U)$~\cite{ChatterjeeRR14}.
For a bound $c\in \mathbb{N}$ the bounded operator
$\textsf{CpreMin}_{c}:\mathscr{A}{(SE(c))}\rightarrow{\mathscr{A}(SE(c))}$,
where $\mathscr{A}(SE(c))$ is the set of antichains of $SE(c)$, is
defined by:
\begin{equation}
\epsilon_{min}(\Lambda) = \{(s_{0},
e_{0})\in{SE(c)}\mid{s_{0}}\in{States_{0}}\wedge{e_{0}
= \min_{(s_{0}, s)\in{E}~s.t.~\exists{e_{1} (s,e_{1})\in{\Lambda}}}
\lbrack~\max(0,e_{1}-w(s_{0},s))~\rbrack}\} \label{eq:CpreMinDefSys}
\end{equation}
%
\begin{equation}
\eta_{min}(\Lambda) = \{(s_{1},
e_{1})\in{SE(c)}\mid{s_{1}}\in{States_{1}}\wedge{e_{1} = \max_{(s_{1},
s)\in{E}~s.t.~\exists{e_{0} (s,e_{0})\in{\Lambda}}}
\lbrack~\max(0,e_{0}-w(s_{1},s))~\rbrack}\} \label{eq:CpreMinDefEnv}
\end{equation}
%
\begin{equation} 
\textsf{CpreMin}_{c}(\Lambda) =
\textsf{Min}(\preceq,\epsilon_{min}(\Lambda)\cup{\eta_{min}(\Lambda)}) =
\epsilon_{min}(\Lambda)\cup{\eta_{min}(\Lambda)} \label{eq:CpreMinDefUnion}
\end{equation}
$\textsf{CpreMin}_{c}$ is used to compute the sets of pairs consisting of a
state and the \textit{minimal} initial energy that suffices for player-0 to win for that
state. For every pair of antichains $A,B\in{\mathscr{A}{(SE(c))}}$, we define
the partial order $\sqsubseteq$ such that $A\sqsubseteq{B}$ iff
$\forall{(s_{b},e_{b})\in{B}}~\exists{(s_{a},e_{a})\in{A}}~$ such that
$(s_{a},e_{a})\preceq{(s_{b},e_{b})}$. Note that $A\sqsubseteq{B}$ iff
$\textsf{Up}(\preceq, B)\subseteq\textsf{Up}(\preceq, A)$. 
$\textsf{CpreMin}_{c}$ is a monotone operator over the complete lattice
$(\mathscr{A}{(SE(c))}, \sqsubseteq, \emptyset,
\{(s,0)\mid{s\in{States}}\})$\footnote{For every
$M\subseteq{\mathscr{A}{(SE(c))}}$, $\inf M :=
\textsf{Min}(\preceq,\bigcup{M})$, $\sup M :=
\textsf{Min}(\preceq,\bigcap_{m\in{M}}\textsf{Up}(\preceq, m))$. For 
details and proofs see~\cite{CL01}.}.
Therefore, there exists a \textit{least} fixed point of $\textsf{CpreMin}_{c}$
in this lattice that can be calculated by iterated
applications of $\textsf{CpreMin}_{c}$ to the antichain
$ZS=\{(s,0)\mid{s\in{States}}\}$. $ZS$ assigns every state $0$
initial energy, which is the minimal initial energy that is sufficient for
player-$0$ to win a $0$ steps game. In general, if $\Lambda$ contains states and
minimal energy levels for player-0 to win in $k$ steps, then
$\textsf{CpreMin}_{c}(\Lambda)$ contains those required for $k+1$ steps.

\section{Example Specification: Elevator}
\label{sec:example}

We present an example specification of a controller for an elevator
servicing multiple floors. The environment inputs are pending requests
to floors and the current floor of the cabin. The controller outputs
are commands for moving up, stopping, or moving down. Quantitative
properties of the specification are expressed as weights on transitions.
Negative weights whenever the elevator is not
on the requested floor and positive weights for reaching a requested
floor.

\paragraph{Reactive specification}
The elevator specification is shown in
List.~\ref{lst:elevator}. The environment controls the variable
\texttt{pending} which signals a request of the elevator cabin to a
destination floor given by the variable \texttt{dest\_floor}. The
environment also maintains variables to keep track of the source floor
\texttt{src\_floor} (location of cabin when request arrived) and the
current floor
\texttt{current\_floor}. The system controls the variable \texttt{move}
with the possible moves of the cabin (l.~7).

\begin{figure}[t]
\lstset{language=SMV}
\begin{lstlisting}[label=lst:elevator, caption={A specification for an
elevator controller consisting of environment variables \texttt{VARENV},
system variables \texttt{VAR}, definitions \texttt{DEFINE}, assumptions
\texttt{ASSUMPTION}, and guarantees \texttt{GUARANTEE}}]
VARENV
  pending : boolean;
  src_floor : 0..4;
  dest_floor : 0..4;
  current_floor : 0..4;
VAR
  move : {UP, DOWN, STOP};

DEFINE
  TOP := 4;
  THERE := current_floor = dest_floor;  
  
ASSUMPTION -- set src_floor 
  G (!pending & next(pending) -> next(src_floor) = next(current_floor));
ASSUMPTION -- remember src_floor and dest_floor
  G (pending -> (next(src_floor) = src_floor & next(dest_floor) = dest_floor));
ASSUMPTION -- keep requests pending
  G (pending & !THERE -> next(pending));
ASSUMPTION -- disable requests
  G(THERE & pending -> next(!pending));  
ASSUMPTION -- do go up
  G (move=UP & current_floor < TOP -> next(current_floor) = current_floor + 1);
ASSUMPTION -- do stop
  G (move=STOP -> next(current_floor) = current_floor);
ASSUMPTION -- do go down
  G (move=DOWN & current_floor > 0 -> next(current_floor) = current_floor - 1);
    
GUARANTEE -- don't go up on top floor
  G (current_floor = TOP -> move!=UP);
GUARANTEE -- don't go down on bottom floor
  G (current_floor = 0 -> move!=DOWN);
\end{lstlisting}
\end{figure}

The specification in List.~\ref{lst:elevator} defines two
abbreviations \texttt{TOP} for the top floor and \texttt{THERE} for the
condition that the current floor is the requested destination floor. The
remainder of the specification defines from l.~13 to l.~26
safety constraints for the environment, i.e., assumptions, and from
l.~28 to l.~31 safety constraints for the system, i.e.,
guarantees. Some important assumptions regard the values of variables
storing the source floor: the value of source floor is the current
floor when a request is issued (l.~14) and source floor and destination
floor do not change while requests are pending (l.~16). Requests are
disabled when the destination floor is reached (l.~18). The assumptions
in l.~21 to l.~26 ensure that the environment sets the current floor
according to the \texttt{move} commands of the system.

The safety constraints for the system ensure that the cabin does not
move up on the top floor (l.~29) and that it does not move down on the
bottom floor (l.~31).

\paragraph{Symbolic elevator game graph}
The assumptions and guarantees in List.~\ref{lst:elevator} describe
the game graph for a reactive game in terms of initial states of the
environment $\theta^e$ and system $\theta^s$ and the transitions of the
environment $\rho^e$ and system $\rho^s$.
The conjunction of all assumptions without the temporal operator \op{G}
defines $\theta^e$ while the conjunction of all guarantees without the
temporal operator \op{G} defines $\theta^s$. In the example of
List.~\ref{lst:elevator} we have no such assumptions or guarantees and
thus $\theta^e \equiv \texttt{TRUE} \equiv \theta^s$.
Similarly, the conjunction of all assumptions with the temporal operator
\op{G} defines $\rho^e$ while the conjunction of all guarantees with the
temporal operator \op{G} defines $\rho^s$. The two guarantees in
List.~\ref{lst:elevator}, ll.~28-31 express that $\rho^s$ evaluates to
\texttt{FALSE} for all states with \texttt{current\_floor=TOP \&
move=UP} or \texttt{current\_floor=0 \& move=DOWN}, i.e., these states
are deadlocks for the system.

A state in the game graph is an assignment to all variables, e.g,
$s_1=$(\texttt{pending}=\texttt{true}, \texttt{src\_floor}=1,
\texttt{dest\_floor}=4, \texttt{current\_floor}=1, \texttt{move=DOWN}).
From this state the assignments to environment variables in successor
states are determined by the transition relation $\rho^e$ and the
assignment to system variable \texttt{move} is restricted by the
transition relation $\rho^s$ to either \texttt{STOP} or \texttt{UP}
leading to two possible successor states of $s_1$.
The number of reachable states in the game graph is 750.
For 10 floors the number of reachable states is 6,000 and for 50 floors
the number of reachable states is 750,000.

\paragraph{Weight definitions}
List.~\ref{lst:weightsClean} shows a definition of weights for
transitions of the elevator. Every weights definition entry
starts with keyword \op{WEIGHT} and the value of the weight followed by
an LTL formula characterizing a set of transitions. As an example, the
first entry in List.~\ref{lst:weightsClean} defines weight 1 for
transitions from states with a pending request and absolute difference
\op{abs}(\texttt{src\_floor}-\texttt{dest\_floor}) = 1 to states with
the request disabled.
As another example, the last entry in
List.~\ref{lst:weightsClean} defines weight -1 for all transitions
from states where a request is pending and the cabin is not at the
destination floor. Intuitively, the weight definition of
List.~\ref{lst:weightsClean} expresses positive weights per
distance traveled.

\begin{figure}
\lstset{language=SMV}
\begin{lstlisting}[label=lst:weightsClean, caption={A weight definition
for transitions of the elevator with reward for reaching a floor
depending on the distance traveled (between 0 and 4) and punishment -1
for not being at the requested floor. Five entries define six weights
-1, 0, 1, 2, 3, and 4}]
WEIGHT: 1 
  pending & next(!pending) & abs(src_floor - dest_floor) = 1;
WEIGHT: 2
  pending & next(!pending) & abs(src_floor - dest_floor) = 2;
WEIGHT: 3
  pending & next(!pending) & abs(src_floor - dest_floor) = 3;
WEIGHT: 4
  pending & next(!pending) & abs(src_floor - dest_floor) = 4;
WEIGHT: -1
  pending & !THERE;
\end{lstlisting}
\end{figure}

We also present an alternative weight definition in
List.~\ref{lst:weightsStd} consisting of two entries only. The first
entry defines weight 1 for all transitions leaving states where a
request is pending and the cabin is at the destination floor (l.~2). The
second entry defines weight -1 for transitions from states with pending
request where the cabin is not at the requested floor.

\begin{figure}
\lstset{language=SMV}
\begin{lstlisting}[label=lst:weightsStd, caption={Weight
definition for transitions of the elevator with reward 1 for
reaching a requested floor and punishment -1 for not being at the
requested floor. Two entries define three weights
-1, 0, 1}]
WEIGHT: 1
  pending & THERE;
WEIGHT: -1
  pending & !THERE;
\end{lstlisting}
\end{figure}

If a transition between states satisfies multiple weight definitions the
values of all weights for the transition are added, e.g., a transition
that satisfies both the formula in line~8 for weight 4 and the formula
in line 10 for weight -1 in List.~\ref{lst:weightsClean} has
weight 3. In case a transition does not satisfy any of the weight
formulas it is assigned the default weight 0, e.g., all transitions from
states where no requests are pending in the weight definition of
List.~\ref{lst:weightsStd}.

The addition of weights for
overlapping sets of transitions and the completion with 0 is
a pre-processing step. This means that the weights definition in
List.~\ref{lst:weightsClean} does not define 5 weights (number of
\op{WEIGHT} entries) but 6 weights (including weight 0 for transitions
defined both in line 2 and line 10). The weights definition in
List.~\ref{lst:weightsStd} defines 3 weights of values -1, 0, and 1.
When referring to the number of weights, e.g., to describe the
complexity of the algorithm the number of weights is the resulting
number of non-overlapping weights and not the number of entries in our
declarative weights specification.

\paragraph{Elevator energy game and initial energy}

The safety constraints of the system and environment in
List.~\ref{lst:elevator} together with the weight definition in
List.~\ref{lst:weightsClean} or List.~\ref{lst:weightsStd} define an
energy game. Intuitively, in an energy game the system starts with a
finite amount of energy and has to make sure that in an infinite play
the accumulated energy does never drop below 0. The accumulated energy
is defined based on the weights of transitions: for every combined step
of the environment and the system (input and output) the weight
definition defines an update of the energy level by adding the weight
value.

For the first weights definition shown in List.~\ref{lst:weightsClean}
the system can win the energy game with a finite initial energy. An
obvious strategy is to always immediately move to a floor once it is
requested by the environment. The accumulated negative weights -1 for
pending requests will then be compensated by the weight for reaching the
destination floor. Interestingly, the elevator can
do even better in terms of minimal energy levels.
The highest minimal required initial energy from any state is 7. One such
worst-case state is (\texttt{pending}=\texttt{true},
\texttt{src\_floor}=4, \texttt{dest\_floor}=4,
\texttt{current\_floor}=1, \texttt{move=DOWN}).
In this state the cabin is on floor 1 but travels to floor 0 while the
environment issued a request to the top floor.\footnote{Note that this
initial state is not required for winning because the system could choose
\texttt{move=UP} but it represents an interesting worst-case
energy level.} The top floor can be reached within 5 steps
(accumulated energy level 7-5=2). The reward for reaching the floor is 0
because \texttt{src\_floor} = \texttt{dest\_floor} (note that this
reward is obtained on outgoing transitions).
The cabin immediately travels to floor 3 where no request is pending
(see assumption in List.~\ref{lst:elevator}, l.~20).
It then arrives at floor 2 still with energy level 2 (because no request
was pending on floor 3) and can now reach all floors within two steps
and thus maintaining at least 0 energy.

For the alternative weights definition shown in
List.~\ref{lst:weightsStd} the system cannot win for any finite
initial energy. From the above example strategy it is clear that with
weight 2 instead of 1 in List.~\ref{lst:weightsStd}, l.~1 the system
could win the energy game\footnote{This non-trivial strategy for winning
with at most weight 2 for 5 floors was indeed found by our
implementation.}.

\section{Our Symbolic Algorithm}
\label{sec:algorithm}

We start with a description
of our \textit{reactive} EG model in terms of the general model of EG,
both presented in Sect.~\ref{sec:preliminaries}, where player-$0$
(maximizer) is the system and player-$1$ (minimizer) is the environment.
Such a description shows that our model is an instance of the general
model despite their differences. We then proceed with the presentation
of a generic version of our symbolic algorithm for reactive EG.

\subsection{Model of Reactive Energy Games}\label{sec:reactiveEG}

Formally, we show a reduction that takes as input a symbolic game
graph $G^{sym} = \langle \theta^e,\theta^s,\rho^e,\rho^s, w\rangle$ for reactive
EG as defined in Sect.~\ref{sec:preliminaries}, and outputs a \textit{bipartite} game graph $\Gamma^{R} = \langle G^{R}=\langle
V^{R},E^{R},w^{R}\rangle, V^{R}_{0}, V^{R}_{1}\rangle$ for EG.
We start by defining two sub graphs of $G^{R}$,
each consists of a simple cycle $C$ such that if any of its vertices is reached
during a play then the players are trapped in $C$
indefinitely. (1) $C^{win}$: a $(+1)$ weight cycle formed by the
vertices $v^{win}_{0}\in{V^{R}_{0}}$, $v^{win}_{1}\in{V^{R}_{1}}$, and the
edges $e^{win}_0=(v^{win}_{0}\, v^{win}_{1})\in{E^{R}}$, $e^{win}_1=(v^{win}_{1}\,
v^{win}_{0})\in{E^{R}}$ such that $w^{R}(e^{win}_0)=1$, $w^{R}(e^{win}_1)=0$. (2)
$C^{loss}$: a $(-1)$ weight cycle formed by the vertices
$v^{loss}_{0}\in{V^{R}_{0}}$, $v^{loss}_{1}\in{V^{R}_{1}}$, and the edges
$e^{loss}_0=(v^{loss}_{0}\, v^{loss}_{1})\in{E^{R}}$, $e^{loss}_1=(v^{loss}_{1}\,
v^{loss}_{0})\in{E^{R}}$ such that $w^{R}(e^{loss}_0)=(-1)$, $w^{R}(e^{loss}_1)=0$.

$\Gamma^{R}$ is constructed by the following two phases: (1) \textit{Takes}
$\langle\theta^e,\theta^s\rangle$ \textit{as inputs:} we add an \textit{initial}
vertex $v_{\emptyset}\in{V^{R}_{1}}$ that corresponds to $\mathscr{V}_{\emptyset}$, i.e., all variables have no values assigned.
In case there are no valid initial states for the environment, i.e., $\theta^e = \emptyset$,
thus all plays on $G^{sym}$ are winning for the system,
we add $e=(v_{\emptyset}, v^{win}_{0})\in{E^{R}}$ with $w^{R}(e)= 0$ leading to $C^{win}$, and output
$\Gamma^{R}$. Otherwise, for every initial state $s^e\in{\theta^e}$ we add a
vertex $v_{s^e}\in{V^{R}_{0}}$ and an edge $e=(v_{\emptyset}, v_{s^{e}})\in{E^{R}}$ such that $w^{R}(e)=0$, whereas for
every $s\in{\theta^e\cap{\theta^s}}$ where $s = s_{e}\cup{s_{s}}$ we add a
vertex $v_{s}\in{V^{R}_{1}}$ and an edge $e=(v_{s^{e}},
v_{s})\in{E^{R}}$ with $w^{R}(e)=0$. For every
$s^{e}\in{\theta^e}$ such that for all $s^{s}\in{\mathscr{V}_{sys}}$, 
$s^{e}\cup{s^{s}}\notin{\theta^s}$, i.e., an initial state which is
a deadlock for the system, we add $e=(v_{s^{e}},v^{loss}_{1})\in{E^{R}}$
with $w^{R}(e)= (-1)$ leading to $C^{loss}$. (2) \textit{Takes}
$\langle\rho^e,\rho^s,w\rangle$ \textit{as inputs:} for every valid environment
transition $t_{s_1, s^{e}_2}\in{\rho^e}$, we
add $v_{s_1}\in{V^{R}_{1}}$, $v_{t_{s_1, s^{e}_2}}\in{V^{R}_{0}}$ and $e =
(v_{s_1}, v_{t_{s_1, s^{e}_2}})\in{E^{R}}$ with $w^{R}(e)=0$. For every
$t_{s_1, s^{e}_2}\in{\rho^e}$ such that for all system choices $s^{s}_2$,
$t_{s_{1},s^{e}_{2},s^{s}_{2}}\notin{\rho^s}$, i.e., it is a deadlock for the
system, we add $e=(v_{t_{s_1, s^{e}_2}},v^{loss}_{1})\in{E^{R}}$ with $w^{R}(e)=
(-1)$ leading to $C^{loss}$. For every
$t_{s_1, s^{e}_2}\in{\rho^e}$ and system choice $s^{s}_2$ such that 
$t_{s_{1},s^{e}_{2},s^{s}_{2}}\in{\rho^s}$, which results in a valid transition
$t_{s_1,s_2}\in{\rho^e\cap{\rho^s}}$ for both players, we add
$v_{s_2}\in{V^{R}_{1}}$ and $e=(v_{t_{s_1, s^{e}_2}},v_{s_2})\in{E^{R}}$ with
$w^{R}(e) = w(s_1,s_2)$. For every state $s_1$ with $v_{s_1}\in{V^{R}_{1}}$ such
that for all environment choices $s^{e}_2$, $t_{s_1, s^{e}_2}\notin{\rho^e}$, i.e., a deadlock state for
the environment, we add $e=(v_{s_{1}},v^{win}_{0})\in{E^{R}}$ with $w^{R}(e)=0$
leading to $C^{win}$. $\Gamma^{R}$ has the following properties:
\begin{lemma}\label{lemma:noInitDeadlocksEnv1}
 Given an initial credit $c\in\mathbb{N}$, $\theta^e = \emptyset$ if, and only
 if for all EG plays $p$ on $\Gamma^{R}$,
 $p=v_{\emptyset}(v^{win}_{0}v^{win}_{1})^{\omega}$, and for all
$k\geq{1}$, $\textsf{EL}(p,k)\geq{0}$.
\end{lemma}
\begin{lemma}\label{lemma:noInitDeadlocksSys}
$s^{e}\in{\theta^e}$ is an initial deadlock state for the system if, and only
if $v_{s^{e}}\in{V^{R}_{0}}$ has one outgoing (-1) weighted edge to $v^{loss}_{1}\in{V^{R}_{1}}$ in ${C^{loss}}$ and
there is no initial energy value $c\in{\mathbb{N}}$ sufficient for player-0 to
win from $v_{\emptyset}\in{V^{R}_{1}}$.
\end{lemma}
\begin{lemma}\label{lemma:noInitDeadlocksEnv2}
$s\in{\theta^e\cap{\theta^s}}$ is an initial deadlock state for the environment
in $G^{sym}$ if, and only if $v_{s}\in{V^{R}_{1}}$ has one outgoing 0
weighted edge to $v^{win}_{0}\in{V^{R}_{0}}$ in ${C^{win}}$ and
every initial energy value $c\in{\mathbb{N}}$ suffices for player-0 to win from
$v_{s}$.
\end{lemma}
\begin{lemma}\label{lemma:noDeadlocksSys}
A valid transition from $s_1$ for environment choice $s^{e}_2$, $t_{s_1,
s^{e}_2}\in{\rho^e}$, is a deadlock for the system if, and only if $v_{t_{s_1,
s^{e}_2}}\in{V^{R}_{0}}$ has one outgoing (-1) weighted edge to $v^{loss}_{1}\in{V^{R}_{1}}$ in ${C^{loss}}$ and
there is no initial energy value $c\in{\mathbb{N}}$ sufficient for player-0 to
win from $v_{s_1}\in{V^{R}_{1}}$.
\end{lemma}
\begin{lemma}\label{lemma:noDeadlocksEnv}
$s_2$ is a deadlock state for the environment in
$G^{sym}$ with a predecessor state $s_1$ such that
$t_{s_1,s_2}\in{\rho^e\cap{\rho^s}}$ if, and only if $v_{s_2}\in{V^{R}_{1}}$ has
one outgoing 0 weighted edge to $v^{win}_{0}\in{V^{R}_{0}}$ in ${C^{win}}$ and
every initial energy value $c\in{\mathbb{N}}$ suffices for player-0 to win from
$v_{s_2}$.
\end{lemma}
\begin{lemma}\label{lemma:initStates}
A play on $G^{sym}$ starts at $s\in{\theta^e\cap{\theta^s}}$
if, and only if a play on $\Gamma^{R}$ starts with the traversal of
$e_1=(v_{\emptyset}, v_{s^{e}})\in{E^{R}}$ by player-$1$ and
$e_0=(v_{s^{e}}, v_{s})\in{E^{R}}$ by player-$0$, with
$w^{R}(e_1)=w^{R}(e_0)=0$.
\end{lemma}
\begin{lemma}\label{lemma:transitions}
A valid $W\in{\mathbb{Z}}$ weighted transition for both players from $s_1$ to
$s_2$, i.e., $t_{s_1,s_2}\in{\rho^e\cap{\rho^s}}$, is taken at step
$j\in{\mathbb{N}}$ of a play $p^{sym}$ on $G^{sym}$ if, and only if at step $2j+2$ of a
play $p$ on $\Gamma^{R}$ player-$1$ traverses $e_1=(v_{s_1}, v_{t_{s_1,
s^{e}_2}})\in{E^{R}}$ and player-$0$ traverses $e_0=(v_{t_{s_1,
s^{e}_2}},v_{s_2})\in{E^{R}}$ such that $w^{R}(e_1)=0$ and $w^{R}(e_0)=W$.
\end{lemma}
\begin{theorem}\footnote{\label{footnote:proofOmission}We omit the proof from
this submission.} Given an initial credit $c\in\mathbb{N}$, for all
$j\geq{1}$, $e\in \mathbb{Z}$, there exists an infinite reactive EG play
$p^{sym}=s_{0}s_{1}\ldots{s_{j-1}}{s_{j}}\ldots$ on $G^{sym}$ with
$\textsf{EL}(p^{sym},j)=e$ if, and only if there exists an EG play\newline
$p = v_{\emptyset}v_{s^{e}_{0}}v_{s_0}v_{t_{s_{0},s^{e}_{1}}}v_{s_1}\ldots{v_{s_{j-1}}{v_{t_{s_{j-1},s^{e}_{j}}}}{v_{s_j}}}\ldots
=
v_{0}v_{1}v_{2}v_{3}v_{4}\ldots{v_{2j}}{v_{2j+1}}{v_{2j+2}}\ldots$
on $\Gamma^{R}$ with
$\textsf{EL}(p,2+2j)=e$.
\label{theorem:modelEquivalence1}
\end{theorem}
\begin{theorem}$^{\ref{footnote:proofOmission}}$
Given an initial credit $c\in\mathbb{N}$, for all $e\in
\mathbb{N}$, $n\in\mathbb{N}$, $1\leq{j}\leq{n}$, there exists a finite winning
reactive EG play $p^{sym}$ on $G^{sym}$ that ends with a deadlock for the
environment, $p^{sym}=s_{0}s_{1}\ldots{s_{n-1}}{s_{n}}$,
with $\textsf{EL}(p^{sym},j)=e$ if, and only if there
exists an EG play $p =
v_{\emptyset}v_{s^{e}_{0}}v_{s_0}v_{t_{s_{0},s^{e}_{1}}}v_{s_1}\ldots{v_{s_{n-1}}{v_{t_{s_{n-1},s^{e}_{n}}}}{v_{s_n}}}({v^{win}_{0}}{v^{win}_{1}})^{\omega}
\newline{=}
v_{0}v_{1}v_{2}v_{3}v_{4}\ldots{v_{2n}}{v_{2n+1}}{v_{2n+2}}({v^{win}_{0}}{v^{win}_{1}})^{\omega}$ on $\Gamma^{R}$ with
$\textsf{EL}(p,2+2j)=e$ and for all
$k\geq{1}$, $\textsf{EL}(p,k)\geq{0}$.
\label{theorem:modelEquivalence2}
\end{theorem}
\begin{theorem}$^{\ref{footnote:proofOmission}}$
Given an initial credit $c\in\mathbb{N}$, there exists a finite losing
reactive EG play $p^{sym}$ on $G^{sym}$ that ends with $s^{e}\in{\theta^e}$
which is a deadlock for the system, $p^{sym}=s^{e}$, if, and only if there
exists an EG play $p =
v_{\emptyset}v_{s^{e}}({v^{loss}_{1}}{v^{loss}_{0}})^{\omega}$ on $\Gamma^{R}$, and
there exists $k\geq{1}$ such that $\textsf{EL}(p,k)<{0}$.
\label{theorem:modelEquivalence3}
\end{theorem}
\begin{theorem}$^{\ref{footnote:proofOmission}}$
Given an initial credit $c\in\mathbb{N}$, for all $e\in
\mathbb{Z}$, $n\in\mathbb{N}$, $1\leq{j}<{n}$, there exists a finite losing
reactive EG play $p^{sym}$ on $G^{sym}$ that ends with a deadlock for the
system, $p^{sym}=s_{0}s_{1}\ldots{s_{n-1}}{s^{e}_{n}}$,
with $\textsf{EL}(p^{sym},j)=e$ if, and only if there
exists an EG play $p =
v_{\emptyset}v_{s^{e}_{0}}v_{s_0}v_{t_{s_{0},s^{e}_{1}}}v_{s_1}\ldots{v_{s_{n-1}}{v_{t_{s_{n-1},s^{e}_{n}}}}}({v^{loss}_{1}}{v^{loss}_{0}})^{\omega}
=$
$v_{0}v_{1}v_{2}v_{3}v_{4}\ldots{v_{2n}}{v_{2n+1}}({v^{loss}_{1}}{v^{loss}_{0}})^{\omega}$
with $\textsf{EL}(p,2+2j)=e$, and there exists
$k\geq{1}$ such that $\textsf{EL}(p,k)<{0}$.
\label{theorem:modelEquivalence4}
\end{theorem}

From Lem.~\ref{lemma:noInitDeadlocksEnv1}~-~\ref{lemma:noDeadlocksEnv} we get
that $\Gamma^{R}$ has no deadlocks, and conclude from
Thm.~\ref{theorem:modelEquivalence1}~-~\ref{theorem:modelEquivalence4} that
every EG play on $\Gamma^{R}$ is infinite, therefore compliant with the general
model of EG. Moreover, by
Thm.~\ref{theorem:modelEquivalence1} every
play on $\Gamma^{R}$ in which neither $C^{loss}$ nor $C^{win}$ are visited,
determines a reactive play on $G^{sym}$ with the \textit{same} energy levels,
and vice versa. From Thm.~\ref{theorem:modelEquivalence2} and
Lem.~\ref{lemma:noInitDeadlocksEnv1},~\ref{lemma:noInitDeadlocksEnv2},~\ref{lemma:noDeadlocksEnv}
we get that every EG (winning) play on $\Gamma^{R}$ that visits $C^{win}$ determines a finite
(winning) play on $G^{sym}$ that ends with a deadlock for the environment. From
Thm.~\ref{theorem:modelEquivalence3},~\ref{theorem:modelEquivalence4} and
Lem.~\ref{lemma:noInitDeadlocksSys},~\ref{lemma:noDeadlocksSys} we get that every EG play on
$\Gamma^{R}$ that visits $C^{loss}$ is losing for the system, and it determines a finite losing game on $G^{sym}$ that ends with a deadlock for the system.
Thm.~\ref{theorem:modelEquivalence1}~-~\ref{theorem:modelEquivalence4} focus on
the energy levels of player-1's vertices in $\Gamma^{R}$, which are the ones of
interest for our reactive EG, and imply that an initial energy level $c\in\mathbb{N}$ suffices for player-0 to win from
$v_{s}\in{V^{R}_1}$ in $\Gamma^{R}$ iff $c$ suffices for player-0 to win from
a state $s$ which is reachable from any of the valid initial states of
$G^{sym}$.

\subsection{Generic Version of Our Algorithm} Alg.~\ref{alg:generic}
presents
 a generic version of our symbolic algorithm for reactive EG, which takes as
 input an energy bound
$\texttt{maxEng}\in{\mathbb{N}}$. It performs within the while loop in
line~\ref{alg:generic:whileFixpoint} a least fixed point calculation of
$\textsf{CpreMin}_{c}$ operator from Eqn.~(\ref{eq:CpreMinDefUnion}) by
its iterated applications to the antichain
$\{(s,0)\mid{s\in\text{States}}\}$, initially assigned to
\texttt{minEngPred} in line~\ref{alg:generic:initialAntichain}.
However, the calculation is \textit{optimized} for our model by means of
a different formulation of Eqn.~(\ref{eq:CpreMinDefUnion}).
We present $\textsf{CpreMin}_{c}^{OPT}$ for a reactive EG graph $\Gamma^{R}$ as
constructed by the reduction of Sect.~\ref{sec:reactiveEG},
where $SE_{1}^{R}(c):= V_{1}^{R}\times{\{n\in\mathbb{N}\mid n\leq{c}\}}$
for $c\in\mathbb{N}$:
\begin{equation}\small
\textsf{CpreMin}_{c}^{OPT}(\Lambda)= \{(v_{1},
e_{1})\in{SE_{1}^{R}(c)}\mid{e_{1} =
\max_{(v_{1}, v_{0})\in{E^{R}}}\big(\min_{(v_{0}, v'_{1})\in{E^{R}} s.t.
\exists{e'_{1}}(v'_{1},
e'_{1})\in{\Lambda}}\big(\max{(0,e'_{1}-w^{R}((v_{0},
v'_{1})))}\big)}\big)\}\label{eq:CpreMinOPTNonSymbolic}
\end{equation}
The $\textsf{CpreMin}_{c}^{OPT}$ operator from
Eqn.~(\ref{eq:CpreMinOPTNonSymbolic}), applies
Eqn.~(\ref{eq:CpreMinDefSys}) followed by Eqn.~(\ref{eq:CpreMinDefEnv}),
i.e., $\eta_{min}$ is applied to the minimal energy values that have just been calculated by $\epsilon_{min}$ in the current iteration. This optimization utilizes the reactive property of the model, i.e., its turn
alternation that induces a bipartite game graph, by which the initial energy
values of player-$i$'s vertices calculated in iteration $k$ only depend on the
values of player-$j$'s vertices calculated in iteration $k-1$, $i\neq{j}$.
Moreover, for optimizing Eqn.~(\ref{eq:CpreMinDefEnv}), it utilizes the
property of 0 weight for all outgoing edges from all $v\in{V_{1}^{R}}$.
We denote by $A^{OPT}_{i}$ and $A_{i}$ the $i$'th element of the chain
results from the least fixed point computation on $\Gamma^{R}$ of
Eqn.~(\ref{eq:CpreMinOPTNonSymbolic}) and
Eqn.~(\ref{eq:CpreMinDefUnion}), respectively.
Lem.~\ref{lemma:CpreMinOptimizationEnv} formally states that these two
operators are equivalent when applied to $\Gamma^{R}$, while the
number of iterations (i.e., the
chain's length) until a fixed point is reached is smaller by a factor
of 2 for $\textsf{CpreMin}_{c}^{OPT}$.
\begin{lemma}\label{lemma:CpreMinOptimizationEnv}
Given an initial credit $c\in\mathbb{N}$, for all $i\geq{0}$, 
$A_{2i}\cap{SE_{1}^{R}(c)} = A^{OPT}_{i}$.
\end{lemma}
We present in line~\ref{alg:generic:OPTCpre} of Alg.~\ref{alg:generic}
the symbolic formulation of Eqn.~(\ref{eq:CpreMinOPTNonSymbolic}), denoted by
$\textsf{CpreMin}_{c}^{symOPT}$, applied to
$G^{sym}$. It invokes the function $f$ of
Eqn.~\ref{eq:FunctionFCpreMinOPTSymbolic} which handles deadlock cases where $w$
is undefined. We
denote by $A^{symOPT}_{i}$ the $i$'th element of the chain resulting from its iterative application by Alg.~\ref{alg:generic}.

\begin{equation}
f(s,s^{e},s^{s}, e', \text{maxEng}) = \begin{cases}
0\ &\mbox{if $t_{s,s^{e}}\notin{\rho^e}$}\\
\text{maxEng}+1\ &\mbox{if $t_{s,s^e,s^s}\notin{\rho^s}$}\\
\max{\lbrack0, e'- w(s,s^e\cup{s^s})\rbrack} &\mbox{otherwise}
\end{cases}\label{eq:FunctionFCpreMinOPTSymbolic}
\end{equation}

\begin{algorithm}[t]
\caption{Generic symbolic fixed point algorithm for reactive energy
games played on a symbolic game graph $G^{sym}$ with initial energy
bound $\texttt{maxEng}\in
\mathbb{N}$ using function $f$ as defined in
Eqn.~(\ref{eq:FunctionFCpreMinOPTSymbolic}) }
\label{alg:generic}
\begin{algorithmic}[1]
  \STATE \textbf{define} minEngStates, minEngPred \textbf{as}
  $(\text{States} \times\mathbb{N})$
  \STATE minEngPred = $\{(s, 0)\mid{s \in
  \text{States}}\}$\label{alg:generic:initialAntichain}
  \WHILE {minEngStates $\neq$ minEngPred}\label{alg:generic:whileFixpoint}
  \STATE minEngStates = minEngPred
  \STATE{minEngPred} =
  $\{(s,e)\in{\text{States}\times{\{0,1,..,\text{maxEng}\}}}~|$\newline
  $\displaystyle ~~~~~~~~~~~~~~~~~
  e=\max_{s^{e}\in{\mathscr{V}_{env}}}\Big(
  \min_{s^{s}\in{\mathscr{V}_{sys}} \text{ s.t.
  }\exists{e'}.(s^{e}\cup{s^{s}},e')\in \text{minEngStates}}\big(
  f(s,s^{e},s^{s}, e',\text{maxEng})\big)\Big)$
  \label{alg:generic:OPTCpre}
  \ENDWHILE
  \RETURN minEngStates
 \end{algorithmic}
\end{algorithm}

\begin{theorem}[Correctness and Completeness of
Alg.~\ref{alg:generic}]\label{theorem:genericAlgCorrectness}
Alg.~\ref{alg:generic} computes the minimal energy value for each state within the bound \texttt{maxEng}.
\end{theorem}
\begin{proof}

By Lem.~\ref{lemma:CpreMinOptimizationEnv}, if $A^{OPT}_{i}$
contains states and minimal energy levels for player-0 to win in $2i$ steps of a play on $\Gamma^{R}$,
then $A^{OPT}_{i+1}$ contains those required for $2i+2$ steps. We infer from
Lem.~\ref{lemma:transitions} that it is equivalent to an increment of a
single step in the respective play on $G^{sym}$. Therefore, from
Thm.~\ref{theorem:modelEquivalence1}~-~\ref{theorem:modelEquivalence4} we get that for every $i\geq{0}$,
$e\in{\mathbb{N}}$, and for every state $s$ of $G^{sym}$ and its respective vertex
$v_{s}\in{V_{1}^{R}\setminus\{v_{\emptyset},v^{win}_{1},v^{loss}_{1}\}}$:
$(s,e)\in{A^{symOPT}_{i}}$ iff $(v_{s},e)\in{A^{OPT}_{i}}$. It shows that line~\ref{alg:generic:OPTCpre} 
implements Eqn.~\ref{eq:CpreMinDefUnion} and thus the algorithm by Chatterjee et al.~\cite{ChatterjeeRR14} 
for the special case of reactive EG.
\end{proof}

\paragraph{Checking Realizability}

Alg.~\ref{alg:generic} computes a set of winning states and their
required minimal initial energy. To check whether the system can win the
energy game we still have to check whether for every initial choice of
the environment described by $\theta^e$ the system has an initial choice
satisfying $\theta^s$ to select a winning state. For winning states $win
= \{s~|~(s, e) \in \texttt{minEngStates}\}$ this check is
$\forall s^{e} ~\exists s^{s}: s^{e} \in \theta^e \Rightarrow s^{e}
\cup s^{s} \in \theta^s \cap win$. The check has a direct
implementation in BDDs and ADDs.

\paragraph{Running Time Complexity}

The running time complexity of Alg.~\ref{alg:generic} in symbolic steps
is in $O(N\cdot\texttt{maxEng})$. The number of iterations of the
while loop is bounded by the number
of states $N$ times the bound \texttt{maxEng}, i.e., the number of
possible unique configurations of the monotonic \texttt{minEngStates}.
A worst case example that does indeed require $N\cdot\texttt{maxEng} +
2$ fixed point iterations --- it changes the value of a single state by
value 1 in every iteration except in the first and last iterations --- is
a cycle of uneven length $N$ (number of states) with weight 1 from
uneven to even states and weight $-1$ otherwise.\footnote{The
worst case behavior here is due to unrealizability. The same game on a
cycle of even length $N$ requires two iterations.}

\section{BDD Algorithm}
\label{sec:algorithmBDD}
We present a BDD-based implementation of Alg.~\ref{alg:generic} to
compute the minimal initial energy for every state.
Alg.~\ref{alg:bdd} takes as input a bound $\texttt{maxEng} \in
\mathbb{N}$, a definition of weights $\texttt{weights}
\subseteq (\mathbb{Z} \times \text{Set of Transitions})$ which is an
implementation of the weight function $w$\footnote{Every weight appears
once with all transitions of that weight.} as defined in
Sect.~\ref{sec:preliminaries}, and returns a relation
$\texttt{minEngStates}\subseteq(\mathbb{N} \times \text{Set of States})$
of pairs of initial required energy and winning states. The result is empty if no state exists that is winning
for initial energy up to \texttt{maxEng}.
States and transitions are implemented as BDDs over the variables
$env\cup{sys}$ and $env\cup{sys}\cup{env'}\cup{sys'}$, respectively.

\begin{algorithm}[t]
\caption{BDD algorithm for minimal required energy per set of
states for environment transitions $\rho^e$,
system transitions $\rho^s$, $\texttt{weights}
\subseteq (\mathbb{Z} \times
\text{Set of Transitions})$, and energy bound $\texttt{maxEng} \in
\mathbb{N}$.}
\label{alg:bdd}
\begin{algorithmic}[1]
  \STATE \textbf{define} minEngStates, minEngPred \textbf{as}
  ($\mathbb{N} \times$ Set of States) 
  \STATE \textbf{add} (0,
  \texttt{TRUE}) \textbf{to} minEngPred
  \WHILE {minEngStates $\neq$ minEngPred}\label{alg:bdd:whileFixpoint}
    \STATE minEngStates = minEngPred \label{alg:bdd:swap}
    \STATE \textbf{empty} minEngPred \label{alg:bdd:empty}
    \STATE remaining = $\{s\in S~|~  (e, S) \in
    \text{minEngStates}\}$
    \label{alg:bdd:remaining}
    \FOR {\textbf{increasing} best $\in \{0\} \cup \{0 \leq e - w \leq
    \text{maxEng}~|~(e, S)
    \in
    \text{minEngStates} \wedge (w, T) \in \text{weights}\}$}
    \label{alg:bdd:iterateBest}
      \STATE \textbf{define} bestT \textbf{as} Transitions
      \STATE bestT = $\emptyset$\label{alg:bdd:bestT}
      \FOR {$(v, T) \in
      \text{weights}$}\label{alg:bdd:iterateWeights} 
      \STATE S = $\{s \in
      S_{e_v}~|~(e_v, S_{e_v}) \in \text{minEngStates} \wedge e_v - v \leq best\}$\label{alg:bdd:states}
        \STATE \textbf{add} T \textbf{transition to} S \textbf{to}
        bestT\label{alg:bdd:addBestT} 
        \ENDFOR\label{alg:bdd:iterateWeightsEnd}
        \STATE B = (\textbf{forceEnvTo} bestT \textbf{transitions})
        $\cap$ remaining\label{alg:bdd:forceEnvTo} 
        \STATE \textbf{add} (best, B)
        \textbf{to} minEngPred \STATE \textbf{remove} B \textbf{from}
        remaining\label{alg:bdd:addBest} 
    \ENDFOR\label{alg:bdd:iterateBestEnd}
  \ENDWHILE
  \RETURN minEngStates
 \end{algorithmic}
\end{algorithm}

Alg.~\ref{alg:bdd} declares two relations of required energy and states
and initializes \texttt{minEngPred} with energy 0 for all possible
states (\texttt{TRUE}).
The main while loop in line~\ref{alg:bdd:whileFixpoint} implements the
fixpoint computation of Alg.~\ref{alg:generic}.
The algorithm stores the result of the last computation in variable
\texttt{minEngStates} and empties the current relation
\texttt{minEngPred} in lines~\ref{alg:bdd:swap} and \ref{alg:bdd:empty}.
The variable \texttt{remaining} is assigned the union of states that can
maintain at least 0 energy in $k$ steps, i.e., the states from
\texttt{minEngStates} (l.~\ref{alg:bdd:remaining}).

In the for-loop from line~\ref{alg:bdd:iterateBest} to
line~\ref{alg:bdd:iterateBestEnd} the algorithm determines for states in
\texttt{remaining} the value \texttt{best} for $k+1$ steps. It computes
the minimal value \texttt{best} such that the system can force the
environment to use only transitions where the required energy of the
successor for $k$ steps minus the weight of the transition is at most
\texttt{best}. The value \texttt{best} is a combination of required
energy values from \texttt{minEngStates} and weights from
\texttt{weights} (see line~\ref{alg:bdd:iterateBest}). The value 0 is
always included and contributes deadlock states of the environment.

The inner for-loop from line~\ref{alg:bdd:iterateWeights} to
line~\ref{alg:bdd:iterateWeightsEnd} collects all combinations of
transitions to target states \texttt{bestT} that require \texttt{best}
energy value or less for predecessor states.
The algorithm iterates over all weight and transition pairs $(v, T)$ and
for every weight selects in line~\ref{alg:bdd:states} successor states
\texttt{S} with accumulated energy $e_v$ such that reaching a selected
successor requires $e_v - v \leq \texttt{best}$ energy for $k+1$ steps. 
In line~\ref{alg:bdd:addBestT} the combination of \texttt{T} transitions
with \texttt{S} target states is added to \texttt{bestT}. This addition
is implemented in BDD operations as $\texttt{bestT} =
\texttt{bestT} \vee (\texttt{T} \wedge prime(\texttt{S}))$.

Finally, the algorithm computes the set \texttt{B} of remaining
predecessor states from which the system can restrict the environment to
take only \texttt{bestT} transitions in line~\ref{alg:bdd:forceEnvTo}.
The method
\textbf{forceEnvTo} \texttt{bestT}
\textbf{transitions} is implemented in BDD operations as $(\rho^e
\Rightarrow (\rho^s \wedge \texttt{bestT})_{\exists sys'})_{\forall
env'}$. Thus, all states in \texttt{B} require at most \texttt{best}
energy for $k+1$ steps. This is also the minimal value because the
algorithm searched by increasing value of \texttt{best}. The states
\texttt{B} are removed from the \texttt{remaining} states.

\paragraph{Running Time Complexity}

The running time complexity of Alg.~\ref{alg:bdd} in symbolic steps is
in $O(N\cdot\texttt{maxEng}^2\cdot|\texttt{weights}|)$. The iterations
of the while loop in line~\ref{alg:bdd:whileFixpoint} are in $O(N\cdot
\texttt{maxEng})$ as discussed in Sect.~\ref{sec:algorithm}.
The iterations of every execution of the for loop in
line~\ref{alg:bdd:iterateBest} are bounded by \texttt{maxEng} because $0
\leq \texttt{best} \leq \texttt{maxEng}$.
Every execution of the innermost for loop executes the loop body
$|\texttt{weights}|$ times (once for every distinct weight). When a
fixpoint is reached, i.e., $\texttt{minEngStates} =
\texttt{minEngPred}$, the algorithm terminates with the last computed
result of \texttt{minEngStates}.

For the example of the elevator specification in
List.~\ref{lst:elevator}, the weight definition of weights from -1
to 4 shown in List.~\ref{lst:weightsClean}, and $\texttt{maxEng}=100$
we have $N=750$ and $|\texttt{weights}|=6$. The algorithm
reaches a fixed point already after 6 iterations of the outer most
while loop in less than a second.

\paragraph{Correctness and completeness}
Alg.~\ref{alg:bdd} implements
Alg.~\ref{alg:generic}. The update operation in
Alg.~\ref{alg:generic}, l.~\ref{alg:generic:OPTCpre} is implemented
by Alg.~\ref{alg:bdd},
ll.~\ref{alg:bdd:remaining}-\ref{alg:bdd:addBest}. The minimum
is implemented by starting with smallest \texttt{best} and
the maximum is implemented by increasing \texttt{best} until
the environment cannot force higher values for \texttt{best}.
Therefore, by Thm.~\ref{theorem:genericAlgCorrectness} it computes the
minimal energy value for each state within the bound.
%

\section{ADD algorithm}\label{sec:algorithmADD}
Alg.~\ref{alg:add} shows our ADD-based implementation of
Alg.~\ref{alg:generic}.
It takes as input a bound $\texttt{maxEng} \in \mathbb{N}$, a function
$\texttt{weights}:\text{Transitions} \rightarrow \mathbb{Z}$ which is an
ADD implementation of the weight function $w$ as defined in
Sect.~\ref{sec:preliminaries}, assigning weights to valid transitions
for both players, and $\rho^e$ and $\rho^s$ as defined in
Sect.~\ref{sec:preliminaries}. It returns a function
$\texttt{minEngStates}:
\text{States}\rightarrow\mathbb{N}\cup\{+\infty\}$ that assigns every
state the minimal required initial energy, or a $+\infty$ value if it is
not winning for initial energy up to \texttt{maxEng}.
States and transitions are implemented as ADDs over the variables
$env\cup{sys}$ and $env\cup{sys}\cup{env'}\cup{sys'}$, respectively.
We denote the value of an ADD $a: \text{Transitions} \rightarrow
\text{Values}$ for a transition from $s_1$ to $s_2 \in$ States by
$a(s_1,s_2)$.

Alg.~\ref{alg:add} declares two functions, \texttt{minEngStates} and
\texttt{minEngPred}, assigning every state the minimal required initial
energy. Line~\ref{alg:add:zeroEngInit} initializes \texttt{minEngPred}
as the constant 0 function for all states, which is the minimal energy
level sufficient for a 0 steps play.
Alg.~\ref{alg:add} also constructs an ADD representing the weighted game
graph, \texttt{arena}, which models invalid transitions, i.e, deadlocks,
by extending \texttt{weights} such that every invalid environment
transition, i.e., $t_{s,s^e}\notin{\rho^e}$, is assigned a $+\infty$ weight (system
wins), and every invalid system transition, i.e., $t_{s, s^e,
s^s}\notin{\rho^s}$, a $-\infty$ weight (system loses).
The \texttt{arena} construction is implemented by two
\textit{if-then-else} operations for ADDs described in
lines~\ref{alg:add:deadlocksInit} and~\ref{alg:add:arenaInit}.

\begin{algorithm}[t] \caption{ADD algorithm for minimal required energy
per state for environment transitions $\rho^e$, system transitions
$\rho^s$, $\texttt{weights} : \text{Transitions} \rightarrow
\mathbb{Z}$, and energy bound $\texttt{maxEng} \in \mathbb{N}$.}
\label{alg:add}
\begin{algorithmic}[1]
  \STATE \textbf{define} minEngStates, minEngPred \textbf{as} (States
  $\rightarrow$ $\mathbb{N}\cup\{+\infty\}$) \STATE \textbf{define}
  arena, deadlocks \textbf{as} (Transitions $\rightarrow$
  $\mathbb{Z}\cup\{-\infty, +\infty\}$) \STATE \textbf{define}
  accumulatedEng \textbf{as} (Transitions $\rightarrow$
  $\mathbb{N}\cup\{+\infty\}$) \STATE deadlocks(s, s') = \textbf{if}
  $\rho^e$(s, s') \textbf{then} $-\infty$
  \textbf{else} $+\infty$\label{alg:add:deadlocksInit}
  \STATE arena(s, s') = \textbf{if} $(\rho^e\cap{\rho^s})$(s, s')
  \textbf{then} weights(s, s') \textbf{else} deadlocks(s,
  s')\label{alg:add:arenaInit}
  \STATE minEngPred(s) = $0$\label{alg:add:zeroEngInit}
  \WHILE {minEngStates $\neq$ minEngPred}\label{alg:add:whileFixpoint}
    \STATE minEngStates = minEngPred \STATE
    accumulatedEng(s,s') = $\text{minEngStates}(\text{s'})$ $\ominus_{maxEng}$
    arena(s,s')\label{alg:add:applyEngOperator} \STATE minEngPred(s) =
    $\max\limits_{s^{e} \in \mathscr{V}_{env}}\min\limits_{s^{s}
    \in \mathscr{V}_{sys}}$ accumulatedEng(s,
    ${s^{e}}\cup{s^{s}}$)\label{alg:add:maxMinAbstraction}
  \ENDWHILE \RETURN minEngStates
 \end{algorithmic}
\end{algorithm}

The while loop in line~\ref{alg:add:whileFixpoint} performs the same
computation as of Alg.~\ref{alg:generic}, a least fixed point
computation of the operator defined in
Alg.~\ref{alg:generic}, l.~\ref{alg:generic:OPTCpre}, but formulates it as a function rather than a relation due to its
implementation by ADD operations.
Each iteration assigns the resulted function of the previous computation for $k$ steps to \texttt{minEngStates},
and computes an updated function \texttt{minEngPred} for $k+1$ steps by three
ADD operations. The first operation in line~\ref{alg:add:applyEngOperator} whose output is assigned to
\texttt{accumulatedEng} implements the innermost arithmetic computation
of Alg.~\ref{alg:generic}, l.~\ref{alg:generic:OPTCpre}. It computes the
energy levels of states and transitions to
successor states in \texttt{minEngStates}\footnote{Technically, to
refer to \texttt{minEngStates} as successor states its ADD is
primed, i.e., the $sys$ and $env$ variables of
\texttt{minEngStates} are replaced by their primed counterparts.}.
The operator
$\ominus_{maxEng}$ is an implementation of
Eqn.~\ref{eq:FunctionFCpreMinOPTSymbolic} compliant with the deadlocks representation in \texttt{arena}, defined as:
\vspace{-0.2em}
\begin{equation}
f_{1}(s') \ominus_{maxEng}f_{2}(s,s') = \begin{cases}
0\ &\mbox{if $f_{2}(s,s')=+\infty$}\\
+\infty\ &\mbox{if $f_{1}(s')=+\infty$ or
$f_{2}(s,s')=-\infty$}\\
+\infty\ &\mbox{if $f_{1}(s') - f_{2}(s,s') > maxEng$}\\
\max{\lbrack0, f_{1}(s') - f_{2}(s,s')\rbrack} &\mbox{otherwise}
\end{cases}\label{eq:addOminusEngOp}
\end{equation}
In iteration $k$, for each source and target states,
Eqn.~(\ref{eq:addOminusEngOp}) results in the required accumulated
energy level for $k$ steps, that is the subtraction of the corresponding transition's weight from the target state's minimal required
energy for $k-1$ steps. If the transition's weight is $+\infty$,
i.e., winning for the system, it results in a 0 value. Otherwise, unless it
exceeds the \texttt{maxEng} bound, which results in a $+\infty$ value (i.e., this transition is losing for the given bound in $k$ steps), a $\max$ operator (with 0 as its second argument) ensures the resulted accumulated energy is non-negative.
Line~\ref{alg:add:maxMinAbstraction} implements the $\min$ operator
followed by the outermost $\max$ operator of Alg.~\ref{alg:generic}, l.~\ref{alg:generic:OPTCpre}
by two ADD abstractions of \texttt{accumulatedEng} for successor
system and environment choices\footnote{Technically, the ADD
minimum abstraction is done on variables $sys'$ and the maximum
abstraction on variables $env'$.}.

\paragraph{Running Time Complexity} 
The running time complexity of
Alg.~\ref{alg:add} in symbolic steps is as the number of iterations of the while loop in
line~\ref{alg:add:whileFixpoint} which is in $O(N\cdot\texttt{maxEng})$, as
discussed in Sect.~\ref{sec:algorithm}.
For the elevator specification in
List.~\ref{lst:elevator}, weights from List.~\ref{lst:weightsClean}, and
$\texttt{maxEng}=100$ the algorithm
reaches a fixed point after 6 iterations (same number of iterations as
the BDD implementation in Alg.~\ref{alg:bdd}).

\paragraph{Correctness and Completeness} Alg.~\ref{alg:add} formulates
Alg.~\ref{alg:generic} in terms of functions rather than relations.
Therefore, by Thm.~\ref{theorem:genericAlgCorrectness} it computes the
minimal energy value for each state within the bound.

\section{Implementation and Preliminary Evaluation}
\label{sec:evaluation}

We have implemented Algorithms~\ref{alg:bdd} and~\ref{alg:add} in Java,
based on an extended version of JTLV~\cite{PnueliSZ10}. We have extended JTLV
with support for ADDs as provided by CUDD. We use CUDD via JNI.

We show a comparison of the two algorithms for computing the minimal
required energy from all winning states, using the
elevator example from Sect.~\ref{sec:example}. The questions our
evaluation addresses are:
\begin{packed_itemize}
  \item How do the ADD and BDD algorithms scale for
  elevator specifications with increasing floors?
  \item Do they scale differently for the weight specification with many
  or few weights?
  \item What is the impact of the choice of the bound \texttt{maxEng} on
  the running time?
\end{packed_itemize}

We run all experiments on an ordinary PC, Intel i7 CPU 3.4GHz, 16GB RAM
with Windows 7 64-bit OS. Our implementation uses CUDD 2.5.0 compiled for 32Bit.
The algorithms are implemented in a modified version of JTLV in Java 7 32Bit.
Our implementations are not distributed and use only a single core of the CPU.
We have run the algorithms on every specification 12 times and report average
times if not stated otherwise.
All times are wall-clock times as measured by Java.\footnote{The automatic variable
reordering of CUDD made running times extremely unpredictable with time
differences up to a factor of five when executing the same synthesis problem. 
Therefore, we have disabled any reordering of BDD and ADD variables.
Both implementations use the order of variables as they appear in the
specification. Without reordering the running times we measured appeared
to be more stable.}

\paragraph{Increasing Number of Floors}
\label{sec:increasingFloors}

To measure how the BDD and ADD algorithms scale on the elevator
specification when increasing the number of floors, we have created
specifications similar to the one shown in List.~\ref{lst:elevator}
with increased numbers of floors in steps of 5 from 5 to 50 floors. The
specification for $n$ floors differs from List.~\ref{lst:elevator} in
ll.~3-5 in an updated range of the floor variables to \texttt{0..$n-1$}
and in l.~10 updated to the new top floor \texttt{TOP := $n-1$}.

For our experiments we have used the weight definition in
List.~\ref{lst:weightsClean} with positive weight values of the
distance traveled to each requested floor. The updated weight
definition for $n$ floors has $n$ entries and defines $n+1$ weights from $-1$ to
$n-1$.
For all experiments we have chosen the bound \texttt{maxEng} of the
maximal initial energy per state to be 100, independent of the number of
floors. This bound makes the elevator specification realizable for all
numbers of floors in the experiment.

\begin{figure}
  \centering
  \includegraphics[width=.44\textwidth]{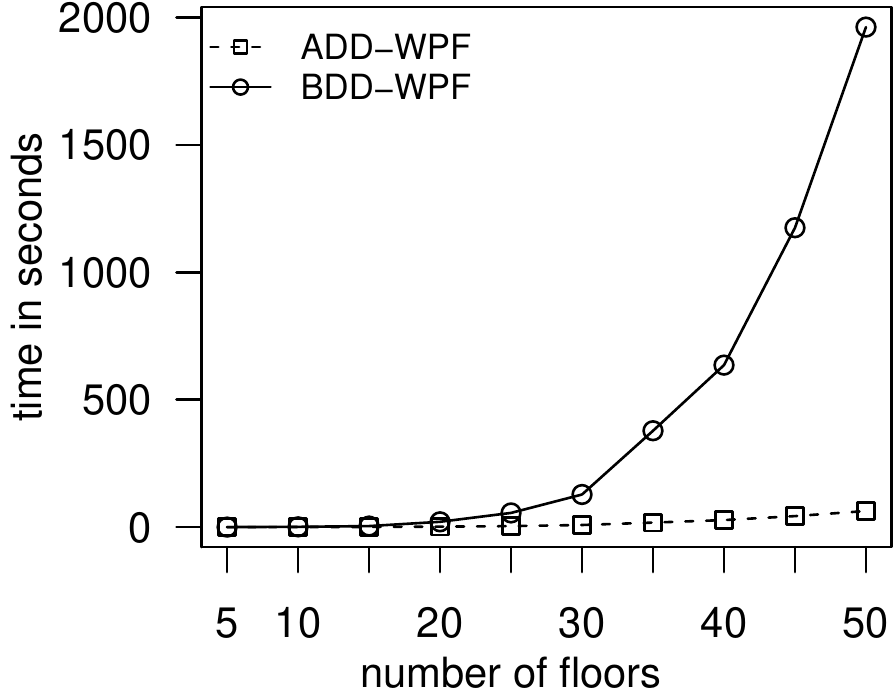}
  ~~~~~~~~
  \includegraphics[width=.44\textwidth]{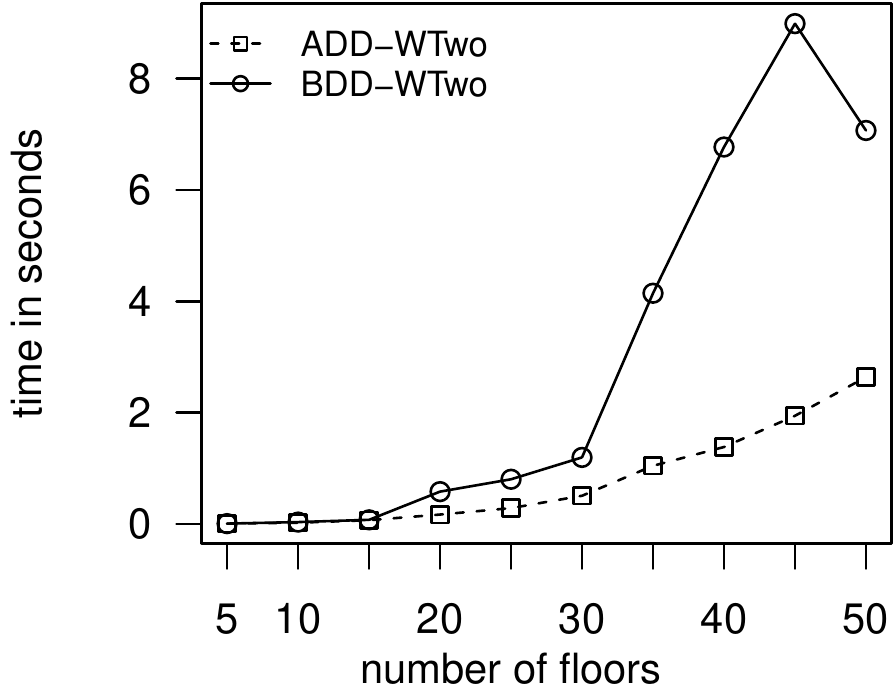}
  \caption{{Running times of ADD and BDD algorithms on elevator
  specification from List.~\ref{lst:elevator} with increasing number
  of floors from 5 to 50 and the energy bound $\texttt{maxEng} = 100$.
  \textbf{Left:} Weight definition of List.~\ref{lst:weightsClean}
  with weights per floor (WPF) adapted to 5 to 50 floors.
  \textbf{Right:} Weight definition of List.~\ref{lst:weightsStd}
  with two weight entries (WTwo) adapted to 5 to 50 floors.
  }}
  \label{fig:increasingFloors}
\end{figure}

The average running times for each specification
are summarized in Fig.~\ref{fig:increasingFloors}~(left). For each
number of floors we present the data points for the ADD algorithm and the
BDD algorithm, measured in seconds. Up to 10 floors the running time is
below a second. The running time for 50 floors of the ADD algorithm is
around one minute while the BDD algorithm runs for 33 minutes. 

From Fig.~\ref{fig:increasingFloors}~(left) it is very clear 
that the ADD algorithm scales much better than the BDD algorithm for elevator
specifications with increasing number of floors.

\paragraph{Different Weight Specifications}

To compare the running time of both algorithms for different weight
specifications we have executed the same experiment as in
Sect.~\ref{sec:increasingFloors} with the alternative weight definition
shown in List.~\ref{lst:weightsStd}. For the elevator with $n$
floors we have changed the single positive weight for reaching a
requested floor in List.~\ref{lst:weightsStd}, l.~1 to $n$. In this
experiment the number of floors ranges again from 5 to 50 in steps of 5
and the number of weights is always 3 ($-1$, $0$, and $n$).

The average running times for each specification
are summarized in Fig.~\ref{fig:increasingFloors}~(right). Again,
the ADD algorithm performed much better than the BDD algorithm.
When comparing absolute running times shown in
Fig.~\ref{fig:increasingFloors}~(left) and (right) it is clear that a different
weight specification for the same synthesis problem can make a significant
difference in running times. For 50 floors the difference in
running times between 3 weights defined in List.~\ref{lst:weightsStd}
and 51 weights defined in List.~\ref{lst:weightsClean} is of factor
277 for the BDD algorithm and of factor 24 for the ADD algorithm.

\paragraph{Choice of Bound \texttt{maxEng}}

The two algorithms we have implemented are both bounded by the maximal initial
energy \texttt{maxEng}. We are thus interested in the
influence of the bound on the running times of the algorithms on specifications
of the elevator. We run both the ADD and the BDD algorithm for increasing bounds
\texttt{maxEng}.

\begin{figure}[t]
  \centering
  \includegraphics[width=.9\textwidth]{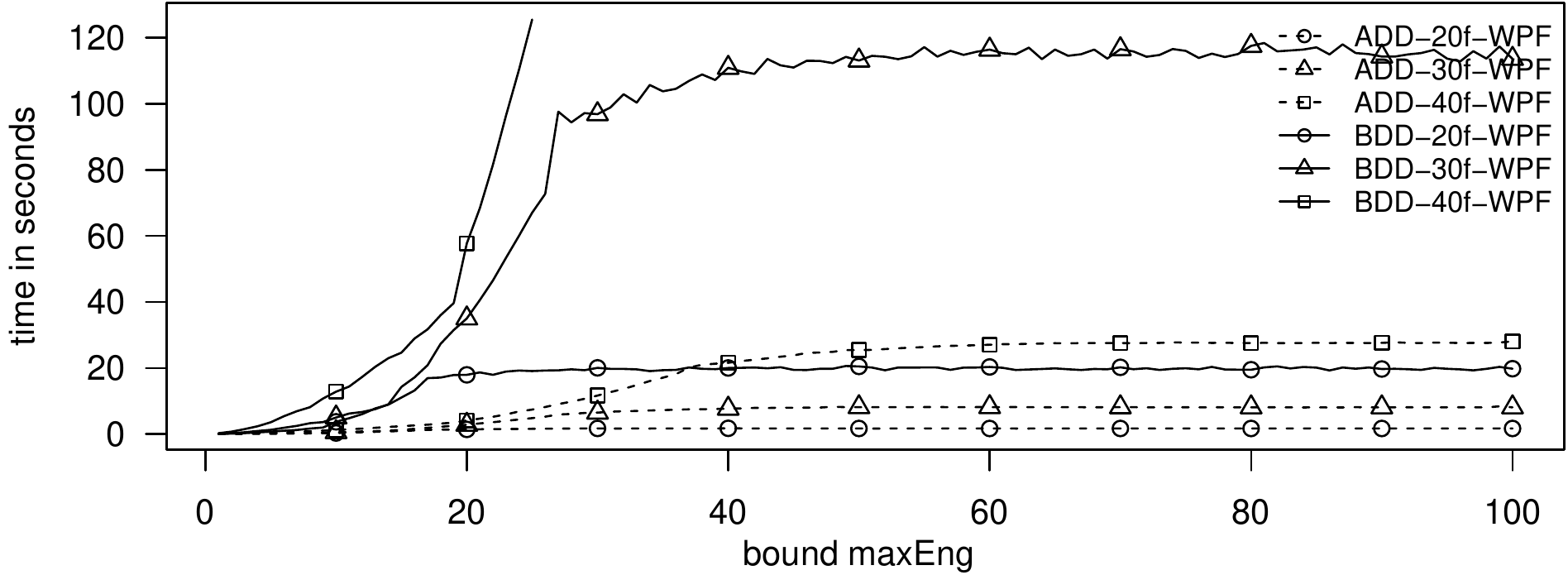}
  \vspace{-.8em}
  \caption{{Running times of ADD and BDD algorithms on
  elevator specification from List.~\ref{lst:elevator} for 20, 30, and 40
  floors. Weight definition from
  List.~\ref{lst:weightsClean} for bounds \texttt{maxEng}
  from 1 to 100 (specification with 20/30/40 floors becomes realizable
  for \texttt{maxEng} bound 36/56/76.}}
  \label{fig:scalePerFloor}
\end{figure}

\begin{figure}[t]
  \centering
  \includegraphics[width=.9\textwidth]{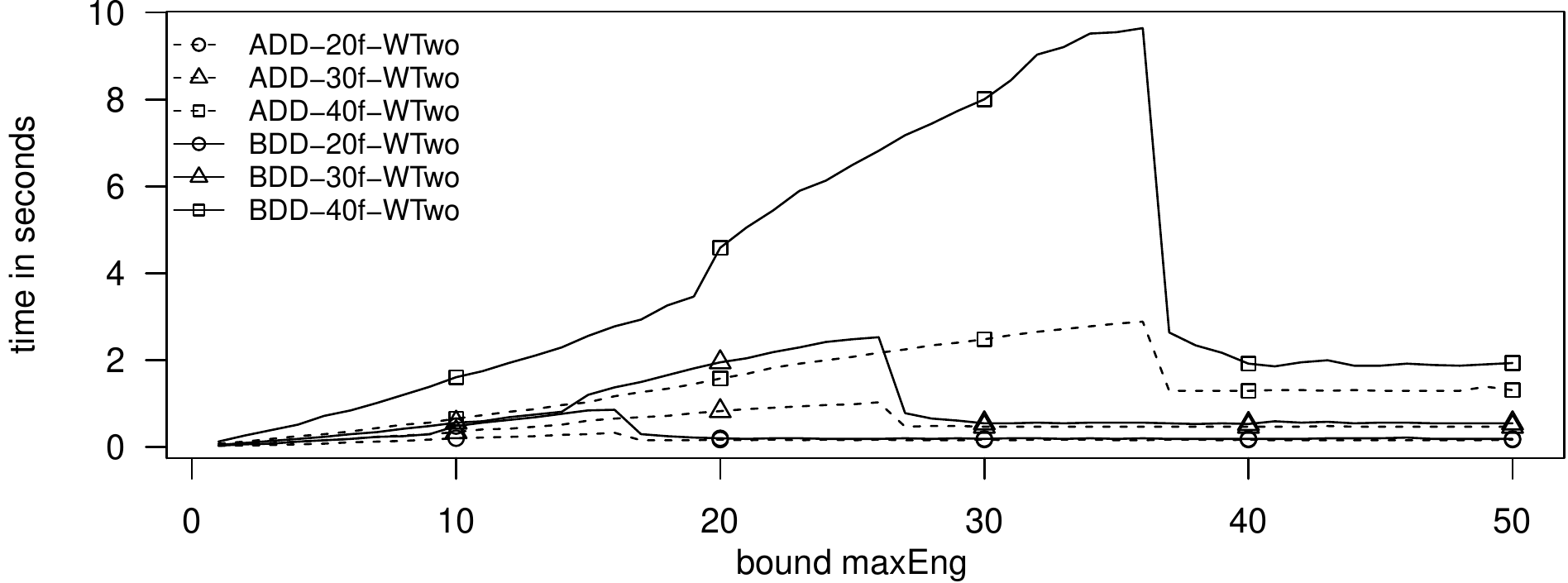}
  \vspace{-.8em}
  \caption{{Running times of ADD and BDD algorithms on
  elevator specification from List.~\ref{lst:elevator} for 20, 30, and 40
  floors. Weight definition
  from List.~\ref{lst:weightsStd} for bounds \texttt{maxEng}
  from 1 to 50 (specification with 20/30/40 floors becomes realizable
  for \texttt{maxEng} bound 19/29/39. }}
  \label{fig:scaleTwo}
\end{figure}

Fig.~\ref{fig:scalePerFloor} shows running times for the elevator
specifications with 20/30/40 floors and the weights definition with
positive weights for the distance traveled to a requested floor from
List.~\ref{lst:weightsClean}. The x-axis shows the bound \texttt{maxEng}
from weight 1 to 100 in steps of 1 (markers every 10 steps to
distinguish 20/30/40 floors). The results are based on a single run for
each specification, algorithm, and bound (600 runs) and thus may contain
some noise, e.g., see the BDD algorithm for 30 floors with
\texttt{maxEng} greater than 20. We have omitted the running times of
the BDD algorithm above 120 seconds (running times go up to 640 seconds
and remain above 600 seconds for increasing bound). The specifications
for 20/30/40 floors are realizable for bounds of at least 36/56/76.

Fig.~\ref{fig:scaleTwo} shows running times for the
elevator specifications with 20/30/40 floors and the weights definition with
three weights from
List.~\ref{lst:weightsStd}. The x-axis shows the bound
\texttt{maxEng} from weight 1 to 50 in steps of 1 (markers every 10
steps to distinguish 20/30/40 floors). The results are
based on a single run for each specification, algorithm, and bound (300
runs). The specifications for 20/30/40 floors are
realizable for bounds of at least 19/29/39. 

From both figures we observe that the ADD and BDD algorithms behave
similarly. We see that the running times become stable a few runs before
the bound for realizability is reached.
For the weight definition of List.~\ref{lst:weightsClean} with more
weights the running times shown in Fig.~\ref{fig:scalePerFloor}
become stable for both algorithms much before the synthesis problem
becomes realizable.
Interestingly, for the weight definition of
List.~\ref{lst:weightsStd} and running times shown in
Fig.~\ref{fig:scaleTwo}, running times increase up to a
factor of five before the problem becomes realizable and then drop again to become
stable.

\section{Discussion and Related Work}
\label{sec:discussion}

The evaluation in Sect.~\ref{sec:evaluation} shows that the ADD
algorithm scales well for increasing the statespace of the synthesis
problem. Both algorithms however show a strong increase in running times
for weight definitions with more weights. The increase in running times
is stronger for the BDD algorithm. 

We have evaluated the two algorithms on specifications with very
restricted game graphs. States with a pending request have a single
environment successor and at most three system choices. These very
sparse game graphs might not be suited to demonstrate the potential of
symbolic computations. Experiments covering a larger set of
specifications are required to better evaluate the potential of our
symbolic implementations. Future analysis should also take
memory consumption into account.


Energy games, as introduced by Bouyer et al.~\cite{BouyerFLMS08}, were
generalized to multi-dimensional energy games by Chatterjee et
al.~\cite{ChatterjeeDHR10}. The energy game algorithm by Brim et
al.~\cite{BrimCDGR11}, which we have extended in our work, has later
been extended by Chatterjee et al.~\cite{ChatterjeeRR14} to also solve multi-dimensional
energy games. Bohy et al.~\cite{BohyBFR13} presented an algorithm for
LTL synthesis via a $k$ bounded reduction through universal $k$-co-Büchi
automata. The algorithm is implemented in Acacia+~\cite{BohyBFJR12} and
is symbolic in the multi-dimensional energy level and bound $k$. Since
their input is LTL and weights are defined on atomic propositions and
not on transitions, our implementations are not easily comparable. A
multi-dimensional extension might be possible for the algorithms we
presented but it appears less natural for the ADD algorithm where
terminals describe minimal energy levels for a single dimension.

Finally, another quantitative game in recent interest are Mean Payoff
Games (MPG), introduced by Ehrenfeucht and Mycielski~\cite{EM79}. Bouyer
et al.~\cite{BouyerFLMS08} showed that MPG and EG are log-space equivalent.
Therefore, by applying a log-space reduction from MPG to EG, our two
proposed symbolic algorithms for EG can also be used to solve MPG.

\section{Conclusion}
\label{sec:conclusion}

We have presented two algorithms for solving reactive energy games.
Given an energy game and a maximal initial energy level both algorithms
compute the minimal required energy per state for the system to maintain
a positive energy level in an infinite play. 
The algorithms are optimized for reactive energy games in the sense that
they update energy levels of system and environment states in one
instead of two steps.
To the best of our
knowledge the algorithms are the first that are fully symbolic both in
the energy levels and in the representation of the underlying safety game. 

Our first algorithm uses BDDs while the second is implemented using
ADDs. Both make specific use of symbolic manipulations for efficiently
performing quantitative operations. We have compared the running times
of the two implementations in a preliminary evaluation and found better
scalability for the ADD algorithm for both extending the statespace and
the number of distinct weight values. Our next step will be to evaluate
the performance of both algorithms on a larger body of specifications,
to implement and evaluate symbolic strategy construction, and to deal
with the unrealizable case.

Our work shows that purely symbolic implementations for
solving energy games are possible and might make quantitative synthesis
more accessible and realistic for reactive systems engineers.

The work is part of a larger project on bridging the gap between the
theory and algorithms of reactive synthesis on the one hand and software
engineering practice on the other.  As part of this project we are
building engineer-friendly tools for reactive synthesis, see,
e.g.,~\cite{MR15patterns,MaozR15synt}.

\paragraph{Acknowledgments}
We thank Yaron Velner for helpful introduction to and discussions about
quantitative games. We thank Jean-Francois Raskin for encouraging
discussions and implementation tips at the Marktoberdorf Summer School
2015. Part of this work was done while SM was on sabbatical as visiting
scientist at MIT CSAIL. This project has received funding from the
European Research Council (ERC) under the European Union's Horizon 2020
research and innovation programme (grant agreement No 638049, SYNTECH).

\bibliographystyle{eptcs}
\bibliography{doc}

\end{document}